\newif\ifarxiv
\arxivtrue
\ifarxiv{}
  \documentclass[a4paper,UKenglish,cleveref,pagebackref,thm-restate,numberwithinsect,colorlinks,nameinlink,citecolor=green!50!black,linkcolor=red!50!black]{lipics-v2019}
\fi{}

\ifarxiv{}
  \nolinenumbers
  \hideLIPIcs
\fi{}

\usepackage[utf8]{inputenc}
\usepackage[T1]{fontenc}
\usepackage{xcolor}
\usepackage{amssymb,amsmath}
\ifarxiv{}
  \usepackage{amsthm}
\fi{}
\usepackage{dsfont} %
\usepackage{graphicx}
\usepackage{booktabs}
\usepackage{tabularx}
\usepackage{multicol}
\usepackage[ruled,vlined,linesnumbered]{algorithm2e}
\ifarxiv{}
\fi{}

\ifarxiv{}
  \usepackage[numbers,sort&compress]{natbib} 
  \usepackage{csquotes}
\fi{}
\usepackage{tikz}

\usepackage{mathtools}

\usepackage{xargs} 
\newcommandx{\set}[2][1=1]{\ensuremath{\{#1,\ldots,#2\}}}
\newcommandx{\tlog}[3][1=,3=]{\log_{#1}^{#3}(#2)}

  \newtheorem{observation}{Observation}
  
  \theoremstyle{definition}
  \newtheorem{rrule}{Reduction Rule}

  \newtheorem{construction}{Construction}
\crefname{observation}{Observation}{Observations}
\crefname{rrule}{Reduction Rule}{Reduction Rules}
\crefname{construction}{Construction}{Constructions}
\Crefname{proposition}{Prop.}{Props.}
\crefname{proposition}{Proposition}{Propositions}
\Crefname{theorem}{Thm.}{Thm.}
\crefname{theorem}{Theorem}{Theorems}
\Crefname{corollary}{Cor.}{Cors.}
\crefname{corollary}{Corollary}{Corollaries}

\newcommand{\calF}{\mathcal{F}}

\newcommand{\I}{\mathcal{I}}
\newcommand{\yes}{\texttt{yes}}
\newcommand{\no}{\texttt{no}}
\newcommand{\RD}{$(\Rightarrow)\quad$}
\newcommand{\LD}{$(\Leftarrow)\quad$}

\newcommandx{\problemdef}[5][2=Input,4=Question]{%
  \begingroup
  \par\medskip
  \noindent \textsc{#1}\nopagebreak[4]
  \par\noindent\hangindent=\parindent\textbf{#2}:  #3\nopagebreak[4]
  \par\noindent\hangindent=\parindent\textbf{#4}:  #5
  \par  \medskip
  \endgroup
}

\newcommand{\N}{\mathbb{N}}
\newcommand{\Nzero}{\mathbb{N}_0}

\renewcommand{\O}{\mathcal{O}}
\renewcommand{\o}{o}
\newcommand{\Ost}{\O^*}

\newcommand{\prob}[1]{\textnormal{\textsc{#1}}}

\newcommand{\msat}{Multistage~2\nobreakdash-SAT}
\newcommand{\msatTsc}{\prob{\msat}}
\newcommand{\msatAcr}{\prob{M2SAT}}
\newcommand{\mqsat}{Multistage~$q$\nobreakdash-SAT}
\newcommand{\mqsatTsc}{\prob{\mqsat}}
\newcommand{\mqsatAcr}{\prob{M$q$SAT}}
\newcommand{\qsat}{$q$\nobreakdash-Satisfiability}
\newcommand{\qsatTsc}{\prob{\qsat}}
\newcommand{\qsatAcr}{\prob{$q$\nobreakdash-SAT}}
\newcommand{\twosat}{$2$\nobreakdash-Satisfiability}
\newcommand{\twosatTsc}{\prob{\twosat}}
\newcommand{\twosatAcr}{\prob{$2$\nobreakdash-SAT}}

\newcommand{\cocl}[1]{\ensuremath{\operatorname{#1}}}
\newcommand{\W}[1]{\cocl{W[#1]}}
\newcommand{\classP}{\cocl{P}}
\newcommand{\NP}{\cocl{NP}}

\newcommand{\coNP}{\cocl{coNP}}
\newcommand{\XP}{\cocl{XP}}
\newcommand{\poly}{\cocl{poly}}

\newcommand{\NPincoNPslashpoly}{\ensuremath{\NP\subseteq \coNP/\poly}}

\newcommand{\pt}{parametric transformation}
\newcommand{\lpt}{linear parametric transformation}

\newcommand{\ETHbreaks}{the ETH breaks}
\newcommand{\ETHbreaksL}{the Exponential Time Hypothesis (ETH) breaks}

\newcommand{\frml}{\mathcal{C}}

\newcommand{\cqed}{\hfill$\diamond$}

\newcommand{\tref}[1]{(\Cref{#1})}
\newcommand{\trefs}[2]{(\Cref{#1,#2})}
\newcommand{\xneg}[1]{\lnot{#1}}
\newcommand{\cls}[2]{\ensuremath{(#1\lor#2)}}
\newcommand{\true}{\ensuremath{\top}}
\newcommand{\false}{\ensuremath{\bot}}
\newcommand{\tfset}{\ensuremath{\{\false,\true\}}}
\renewcommand{\frml}{\ensuremath{\phi}}
\newcommand{\seqf}{\ensuremath{\Phi}}
\newcommand{\assment}{truth assignment}
\newcommand{\restr}[2]{\ensuremath{#1|_{#2}}}

\newcommand{\eps}{\ensuremath{\varepsilon}}

\newcommand{\xcase}[2]{\textit{Case~#1}: #2.}

\newcommand{\xor}{\oplus}

\DeclareMathOperator*{\bigland}{\bigwedge}

\newcommand{\ceq}{\ensuremath{\coloneqq}}
\newcommand{\cif}{\text{if }}
\newcommand{\otw}{\text{otherwise}}

\definecolor{lilla}{HTML}{750787}

\usepackage{etoolbox}

\newcommand{\mytitle}{A Multistage View on 2-Satisfiability}

\ifarxiv{}
  \title{\mytitle} %
  \titlerunning{Multistage 2-SAT} %
  
  \author{Till Fluschnik}
  {Technische Universit\"at Berlin, Faculty~IV, Algorithmics and Computational Complexity, Germany}
  {till.fluschnik@tu-berlin.de}
  {https://orcid.org/0000-0003-2203-4386}
  {Supported by DFG, project TORE (NI 369/18).}%

  \authorrunning{T.~Fluschnik} %

  \Copyright{T.~Fluschnik} %



  \category{} %

  \relatedversion{} %

  \supplement{}%

  \acknowledgements{I thank Hendrik Molter and Rolf Niedermeier for their constructive feedbacks.}%

  \EventEditors{John Q. Open and Joan R. Access}
  \EventNoEds{2}
  \EventLongTitle{42nd Conference on Very Important Topics (CVIT 2016)}
  \EventShortTitle{CVIT 2016}
  \EventAcronym{CVIT}
  \EventYear{2016}
  \EventDate{December 24--27, 2016}
  \EventLocation{Little Whinging, United Kingdom}
  \EventLogo{}
  \SeriesVolume{42}
\fi{}

\begin{document}
\ifarxiv{}\maketitle
\fi{}
\begin{abstract}
We study~\prob{$q$-SAT} in the multistage model, 
focusing on the linear-time solvable~\prob{$2$-SAT}.
Herein,
given a sequence of~$q$-CNF fomulas and a non-negative integer~$d$,
the question is whether there is a sequence of satisfying \assment{}s
such that for every two consecutive \assment{}s,
the number of variables whose values changed is at most~$d$.
We prove that \msatTsc{} is \NP-hard even in quite restricted cases.
Moreover,
we present parameterized algorithms (including kernelization) for~\msatTsc{}
and prove them to be asymptotically optimal.
\end{abstract}

\section{Introduction}

\qsatTsc{} (\qsatAcr{}) is one of the most basic and best studied decision problems in computer science:
It asks whether a given boolean formula
in conjunctive normal form, 
where each clause consists of at most~$q$ literals,
is satisfiable.
\qsatAcr{} is~\NP-complete for~$q\geq 3$,
while \twosatTsc{} (\twosatAcr{})
is linear-time solvable~\cite{AspvallPT79}.
The recently introduced multistage model~\cite{EMS14,GuptaTW14}
takes a sequence of instances of some decision problem
(e.g., 
modeling one instance that evolved over time),
and asks
whether there is a sequence of solutions to them
such that,
roughly speaking,
any two consecutive solutions do not differ too much.
We introduce 
\qsatAcr{}
in the multistage model,
defined as follows.\footnote{We identify \texttt{false} and \texttt{true} with~$\false$ and~$\true$, 
respectively.}

\problemdef{\mqsatTsc{} (\mqsatAcr{})}
{A set~$X$ of variables, 
a sequence~$\seqf=(\frml_1,\dots,\frml_\tau)$,
$\tau\in\N$,
of $q$-CNF formulas over literals over~$X$, 
and an integer~$d\in\Nzero$.}
{Are there~$\tau$ \assment{}s~$f_1,\dots,f_\tau\colon X\to \{\false,\true\}$ such that 
\begin{enumerate}[(i)]
 \item for each~$i\in\set{\tau}$, $f_i$ is a satisfying \assment{} for~$\frml_i$,
and
 \item for each~$i\in\set{\tau-1}$, it holds 
that~$|\{x\in X\mid f_i(x)\neq f_{i+1}(x)\}|\leq d$?
\end{enumerate}
}

\noindent
Constraint~(ii) 
of~\mqsatAcr{}
can also be understood as that the Hamming~distance of two consecutive \assment{}s interpreted as~$n$-dimensional vectors over~$\{\false,\true\}$ is at most~$d$,
or when considering the sets of variables set true,
then the symmetric difference of two consecutive sets is at most~$d$.

In this work,
we focus on~\msatAcr{} yet relate most of our results to~\mqsatAcr{}.
We study \msatAcr{}
in terms of classic computational complexity and parameterized algorithmics~\cite{cygan2015parameterized}.

\ifarxiv{}
  \subparagraph{Motivation.}
\fi{}
In theory as well as in practice,
it is common to model problems
as~\qsatAcr{}- or even~\twosatAcr{}-instances.
Once being modeled, 
established solvers specialized on~\qsatAcr{} are employed.
In some cases,
a sequence of problem instances 
(e.g., modeling a problem instance that changes over time)
is to solve 
such that any two consecutive solutions are similar in some way
(e.g., when costs are inferred for setup changes).
Hence,
when following the previously described approach,
each problem instance is first modeled as a \qsatAcr{} instance such that
a sequence of \qsatAcr{}-instances remains to be solved.
Comparably to the single-stage setting,
understanding the multistage setting could give raise to a general approach for solving different (multistage) problems.
With \mqsatAcr{} we introduce the first problem that models the described setup.
Note that,
though a lot of variants of \qsatAcr{} exist,
\mqsatAcr{} is one of the very few variants that deal with a sequence of \qsatAcr{}-instances~\cite{Ramnath04}.

\ifarxiv{}
  \subparagraph{Our Contributions.}
\fi{}
Our results for~\msatTsc{} are summarized in~\cref{fig:results}.
\begin{figure}[t!]
 \centering
  \begin{tikzpicture}
 
    \usetikzlibrary{patterns,backgrounds,shapes}
    \usetikzlibrary{calc}

    \ifarxiv{}
      \def\yr{1}
      \def\xr{1}
    \fi{}

    \ifarxiv{}
      \def\boxw{2.9}
    \fi{}
    \def\boxh{0.6}

    \def\colXP{yellow!50!orange!20}
    \def\colpNPh{orange!50!red!18!white}
    \def\colFPTnoPK{green!10!white}
    \def\colFPTPK{green!20!white}
    \def\colOpen{blue!20!white}
    \ifarxiv{}
      \def\fsres{\footnotesize}
      \def\fsresx{\scriptsize}
    \fi{}
    \def\pnph{para-\NP-h.}
    
    \tikzstyle{xarc}=[->,gray,thick,>=latex,rounded corners]
    
    \newcommand{\parabox}[7]{
      \node (a#1) at (#2)[rectangle, rounded corners, minimum width=\xr*\boxw cm, minimum height=\yr*\boxh cm,fill=white,very thick]{{\small\boldmath#3}};
      \node (c#1) at (a#1.south)[anchor=north,rectangle, rounded corners, minimum width=\xr*\boxw*1 cm, minimum height=1.4*\yr*\boxh cm,fill=#6,align=center,font=\fsres]{#4\\#5};
      \node (b#1) at (c#1.south)[anchor=north,rectangle, rounded corners, minimum width=\xr*\boxw*1 cm, minimum height=1.45*\yr*\boxh cm,fill=lightgray!15!white,align=center,font=\fsresx]{};
      \draw[rounded corners,fill=lightgray] (b#1.south west) rectangle (a#1.north east);
      \node (a#1) at (#2)[rectangle, rounded corners, minimum width=\xr*\boxw cm, minimum height=\yr*\boxh cm,fill=white,very thick,draw]{{\small\boldmath#3}};
      \node (c#1) at (a#1.south)[anchor=north,rectangle, rounded corners, minimum width=\xr*\boxw*1 cm, minimum height=1.4*\yr*\boxh cm,fill=#6,align=center,font=\fsres,draw]{#4\\#5};
      \node (b#1) at (c#1.south)[anchor=north,rectangle, rounded corners, minimum width=\xr*\boxw*1 cm, minimum height=1.45*\yr*\boxh cm,fill=lightgray!15!white,align=center,font=\fsresx]{#7};
    }

    \def\xsh{3.35}
    \def\ysh{3}
    
    \def\teps{0.5}

    \parabox{m}{0.0*\xr,0*\yr}{$m$}{\pnph{}}{\tref{thm:nphardnessd}}{\colpNPh}{$ m= 6$};
    \parabox{n}{\xsh*\xr,0.0*\ysh*\yr}{$n$}{FPT, no~PK}{\tref{thm:fptn}}{\colFPTnoPK}{$2^{\O(n)}$ $\mid$\\ no~$2^{o(n)}$ $\dagger$};
    \parabox{d}{0.5*\xsh*\xr,-1*\ysh*\yr}{$d$}{\pnph{}}{\tref{thm:nphardnessd}}{\colpNPh}{$d=1$ $\mid$\\ lin.\ time if $d=0$};
    \parabox{nd}{1.5*\xsh*\xr,-1*\ysh*\yr}{$n-d$}{\XP, \W{1}-h.}{\tref{thm:cnc}}{\colXP}{$n^{\O(n-d)}$ $\mid$\\ no~$n^{o(n-d)}$};
    \parabox{tau}{2*\xsh*\xr-0.0*\xr,0*\yr}{$\tau$}{\pnph{}}{\tref{thm:nphardnesstau}}{\colpNPh}{$\tau=2$ $\mid$\\ lin.\ time if $\tau=1$};

    \draw[xarc] (ad) to (bn);
    \draw[xarc] (and) to (bn);
    
    \parabox{md}{0.5*\xsh*\xr,1*\ysh*\yr}{$m+d$}{\pnph{}}{\tref{thm:nphardnessd}}{\colpNPh}{$m=6$ and~$d=1$};
    \parabox{mnd}{-0.5*\xsh*\xr,1*\ysh*\yr}{$m+n-d$}{FPT, no PK}
    {\tref{thm:fptmnd}}{\colFPTnoPK}{no~$(n-d)^{f(m)}$ PK};
    
    \parabox{taud}{2.5*\xsh*\xr,1*\ysh*\yr}{$\tau+d$}{XP, \W{1}-h.}{\trefs{thm:xptaud}{thm:nphardnesstau}}{\colXP}{$n^{\O(\tau\cdot d)}$ $\mid$\\ no~$n^{o(d)\cdot f(\tau)}$ $\ddagger$};
    \parabox{taund}{1.5*\xsh*\xr,1*\ysh*\yr}{$\tau+n-d$}{XP, \W{1}-h.}{\tref{thm:allbutkNPh}}{\colXP}{no~$n^{o(n-d)\cdot f(\tau)}$};

    \def\xang{13}
    \draw[xarc] (am) to (bmd);
    \draw[xarc] (am) to (bmnd);
    \draw[xarc] (atau) to (btaud);
    \draw[xarc] (atau) to (btaund);
    \draw[xarc] (ad) to (bmd);
    \coordinate (x) at ($(bd.south)+(0,-0.5*\teps)$);
    \draw[xarc] (bd) to (x) to (btaud|-x) to (btaud);
    \coordinate (x) at ($(bnd.south)+(0,-0.8*\teps)$);
    \draw[xarc] (bnd) to (x) to (bmnd|-x) to (bmnd);
    \draw[xarc] (and) to (btaund);

    \parabox{mn}{0*\xr,2*\ysh*\yr}{$m+n$}{FPT, no~PK}{\tref{thm:nopkmn}}{\colFPTnoPK}{no~$n^{f(m,d)}$~PK};
    \parabox{mtau}{\xsh*\xr,2*\ysh*\yr}{$m+\tau$}{FPT, PK}{\tref{thm:qaudkermntau}}{\colFPTPK}{$\O(m \tau)$~PK $\mid$\\ no~$(m+\tau)^{2-\eps}$~PK $\mathparagraph$};
    \parabox{ntau}{2*\xsh*\xr,2*\ysh*\yr}{$n+\tau$}{FPT, PK}{\tref{thm:qaudkermntau}}{\colFPTPK}{$\O(n^2\tau)$~PK $\mid$\\ no~$\O(n^{2-\eps}\tau)$~PK $\mathsection$};

    \draw[xarc] (amd) to (bmn);
    \draw[xarc] (amnd) to (bmn);
    \draw[xarc] (am) to ($(bmn.south)+(0,-\teps)$) to (bmtau);
    \draw[xarc] (atau) to ($(bntau.south)+(0,-\teps)$) to (bmtau);
    \draw[xarc] (an) to ($(bmtau.south)+(0,-\teps)$) to (bmn);
    \draw[xarc] (an) to ($(bmtau.south)+(0,-\teps)$) to (bntau);
    \draw[xarc] (atau) to (btaud);
    \draw[xarc] (ataud) to (bntau);
    \draw[xarc] (ataund) to (bntau);
  \end{tikzpicture}
  \caption{Our results for \protect\msatTsc{}.
  Each box gives, 
  regarding to a parameterization (top layer),
  our parameterized classification (middle layer) with additional details on the corresponding result (bottom layer).
  Arrows indicate the parameter hierarchy: 
  An arrow from parameter~$p_1$ to~$p_2$ indicates that $p_1\leq p_2$.
  ``PK'' and~``no PK'' stand for 
  ``polynomial problem kernel''
  and ``no polynomial problem kernel unless~\NPincoNPslashpoly{}'',
  respectively.
  \quad
  $\dagger$:~unless \ETHbreaks~\tref{thm:nphardnessd}.\quad 
  $\ddagger$:~unless \ETHbreaks~\tref{thm:nphardnesstau}.\quad 
  $\mathparagraph$:~unless~\NPincoNPslashpoly~\tref{thm:nolinkermntau} \quad
  $\mathsection$:~unless~\NPincoNPslashpoly~\tref{thm:nphardnesstau}.%
  }
  \label{fig:results}
\end{figure}
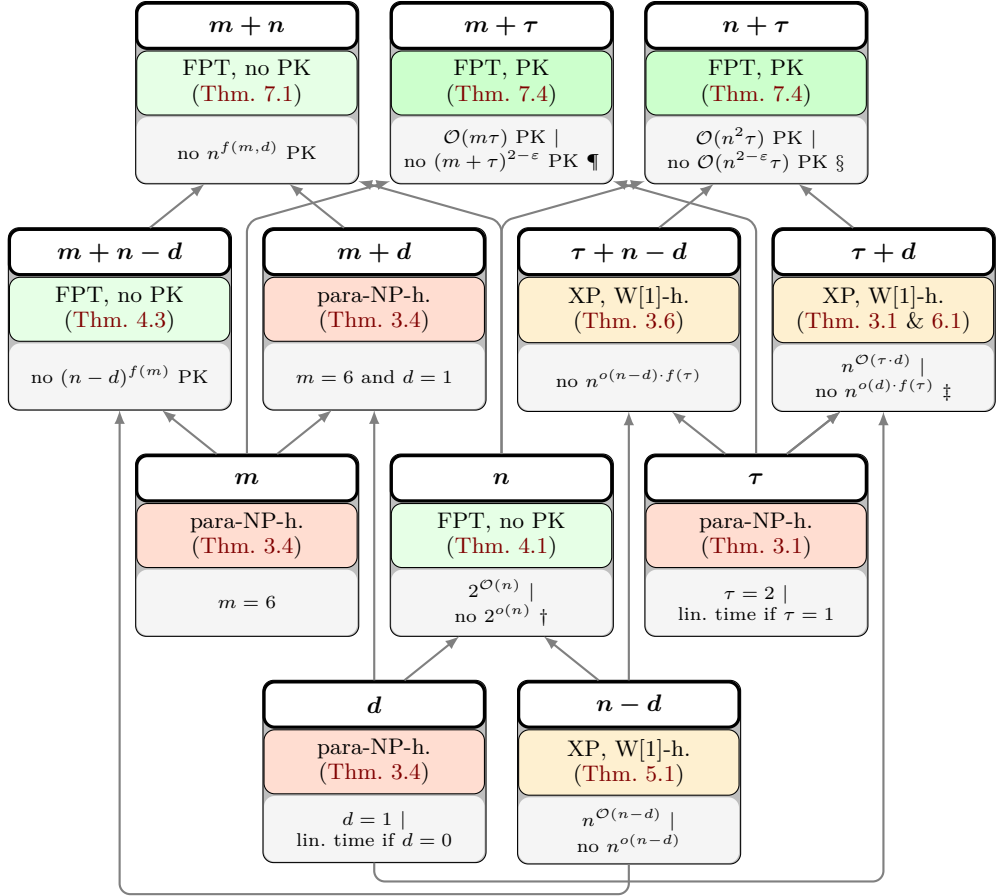
We prove \msatTsc{} to be \NP-hard,
even in fairly restricted cases:
(i)~if~$d=1$ and the maximum number~$m$ of clauses in any stage is six, or
(ii)~if there are only two stages.
These results are tight in the sense that \msatAcr{} is linear-time solvable when~$d=0$ or~$\tau=1$.
While \NP-hardness for~$d=1$ implies that there is no~$(n+m+\tau)^{f(d)}$-time algorithm for any function~$f$ unless~$\classP=\NP$,
where $n$ denotes the number of variables,
we prove that when parameterized by the dual parameter~$n-d$ 
(the minimum number of variables not changing between any two consecutive layers),
\msatAcr{} is~\W{1}-hard and solvable in $\Ost(n^{\O(n-d)})$ time.\footnote{The~$\O^*$-notation suppresses factors polynomial in the input size.}
We prove this algorithm to be tight in the sense that,
unless~\ETHbreaksL,
there is no~$\Ost(n^{o(n-d)})$-time algorithm.
Further,
we prove that \msatAcr{} is solvable in~$\O^*(2^{\O(n)})$~time but not in~$\O^*(2^{\o(n)})$~time unless~\ETHbreaks.
Likewise,
we prove that \msatAcr{} is solvable in~$\O^*(n^{\O(\tau\cdot d)})$~time but not in~$\O^*(n^{o(d)\cdot f(\tau)})$~time for any function~$f$ unless~\ETHbreaks.
As to efficient and effective data reduction,
we prove \msatAcr{} to admit problem kernelizations of size~$\O(m\cdot \tau)$ and~$\O(n^2\tau)$,
but none of size~$(n+m)^{\O(1)}$,
$\O((n+m+\tau)^{2-\eps})$,
or~$\O(n^{2-\eps}\tau)$, 
$\eps>0$, 
unless~\NPincoNPslashpoly.

\ifarxiv{}
  \subparagraph{Related Work.}
\fi{}

\qsatAcr{} is one of the most famous 
decision problems
with a central role 
in \NP-completeness theory~\cite{Cook71,Karp72},
for the (Strong) Exponential Time Hypothesis~\cite{ImpagliazzoP01,ImpagliazzoPZ01},
and in the early theory on kernelization lower bounds~\cite{FortnowS11,BodlaenderDFH09},
for instance.
In contrast to~\qsatAcr{} with~$q\geq 3$,
\twosatAcr{}
is proven to be polynomial-~\cite{Krom67}, 
even linear-time~\cite{AspvallPT79} solvable.
Several applications of \twosatAcr{} are known
(see, e.g., \cite{RaghavanCS86,EvenIS76,ChrobakD99,HansenJ87}).
In the multistage model,
various problems from different fields were studied,
e.g.\ graph~theory~\cite{FluschnikNRZ19,FluschnikNSZ20,ChimaniTW20,GuptaTW14,BampisELP18,BampisEK19},
facility location~\cite{EMS14},
knapsack~\cite{BampisET19},
or committee elections~\cite{BredereckFK20}.
Also variations to the multistage model were studied,
e.g.\ with a global budget~\cite{HeegerHKNRS20},
an online-version~\cite{BampisEST19},
or using different distance measures for consecutive stages~\cite{BredereckFK20,FluschnikNSZ20}.

\section{Preliminaries}
\label{sec:prelims}

We denote by~$\N$ and~$\Nzero$ the natural numbers excluding and including zero,
respectively.
Frequently,
we will tacitly make use of 
the fact that
for every~$n\in\N$, $0\leq k\leq n$, 
it holds true that $1+\sum_{i=1}^k \binom{n}{i}= \sum_{i=0}^k \binom{n}{i}\leq 1+n^k\leq 2n^k$.

\ifarxiv{}
  \subparagraph*{Satisfiability.}
\fi{}
Let~$X$ denote a set of variables.
A literal is a variable that is either positive or negated 
(we denote the negation of~$x$ by~$\xneg{x}$).
A clause is a disjunction over literals.
A formula~$\frml$ is in conjunctive normal form (CNF) if it is of the form~$\bigland_i C_i$,
where~$C_i$ is a clause.
A formula~$\frml$ is in~$q$-CNF if it is in CNF and each clause consists of at most~$q$ literals.
An \assment{}~$f\colon X\to\{\false,\true\}$ is satisfying for~$\frml$ (or satisfies $\frml$) 
if each clause is satisfied,
which is the case if at least one literal in the clause is evaluated to true 
(a positive variable assigned true, 
or a negated variable assigned false).
For~$a,b\in\tfset$,
let~$a\xor b\ceq\false$ if~$a=b$,
and $a\xor b\ceq\true$ otherwise.
For~$X'\subset X$,
an \assment{}~$f'\colon X'\to \tfset$ is called \emph{partial}.
We say that
we \emph{simplify} a formula~$\frml$ given a partial \assment{}~$f'$ 
(we denote the simplified formula by~$\frml[f']$)
if each variable $x\in X'$ is replaced by~$f'(x)$,
and then each clause containing an evaluated-to-true literal is deleted.

\ifarxiv{}
  \subparagraph*{Parameterized Algorithmics.}
\fi{}

A parameterized problem~$L$ is a set of instances~$(x,p)\in\Sigma^*\times\Nzero$,
where~$\Sigma$ is a finite alphabet and~$p$ is referred to as the parameter.
A parameterized problem~$L$ is
(i) fixed-parameter tractable (in~FPT) if each instance~$(x,p)$ can be decided for~$L$ in~$f(p)\cdot |x|^{\O(1)}$ time,
and
(ii) in XP if each instance~$(x,p)$ can be decided for~$L$ in~$|x|^{g(p)}$ time,
where~$f,g$ are computable functions only depending on~$p$.
If~$L$ is~\W{1}-hard,
it is presumably not in~FPT.
A problem bikernelization for a parameterized problem~$L$ to a parameterized problem~$L'$ takes any instance~$(x,p)$ of~$L$
and maps it in polynomial time to an equivalent instance~$(x',p')$ of~$L'$ 
(the so-called problem bikernel) 
such that~$|x'|+p'\leq f(p)$ for some computable function~$f$.
A problem kernelization is a problem bikernelization where~$L=L'$.
If~$f$ is a polynomial, 
the problem (bi)kernelization is said to be polynomial.
A \pt{} from a parameterized problem~$L$ to a parameterized problem~$L'$ maps any instance~$(x,p)$ of~$L$ in $f(p)\cdot |x|^{\O(1)}$~time to an equivalent instance~$(x',p')$ of~$L'$ such that~$p'\leq g(p)$ for some functions $f,g$ each only depending on~$p$.
If there is a \pt{} from~$L$ to~$L'$ with~$L$ being~\W{1}-hard,
then~$L'$ is \W{1}-hard.
If~$f(p)\in p^{\O(1)}$ and~$g(p)\in \O(p)$,
then we have a \lpt{}~\cite{HermelinW12}.
If there is a \lpt{} from a problem~$L$ to problem~$L'$ with~$L'$ admitting a problem kernelization of size~$h(p')$,
then~$L$ admits a problem bikernelization of size~$h(p')$.
\ifarxiv{}
  \subparagraph*{Preprocessing on \msatTsc.}
\fi{}

Due to the following data reduction,
we can safely assume each stage to admit a satisfying \assment{}.

\begin{rrule}
 \label{rrule:trivialno}
 If a stage exists with no satisfying \assment{},
 then return \no.
\end{rrule}

\ifarxiv{}
\noindent
 \cref{rrule:trivialno} is correct and applicable in linear time.
\fi{}

\section{From Easy to Hard: NP- and W-hardness}
\label{sec:hardness}

\msatTsc{} is linear-time solvable if the input consists of only one stage,
or if all or none variables are allowed to change its \assment{} between two consecutive stages.

\begin{observation}%
 \label{obs:msatlintime}
 \msatTsc{} is linear-time solvable if (i)~$\tau=1$, (ii)~$d=0$, 
 or (iii)~$d=n$.
\end{observation}

\begin{proof}
 Let~$(X,\seqf=(\frml_1,\dots,\frml_\tau),d)$ be an instance of~\msatAcr{}.
 \xcase{(i)}{$\tau=1$}
 Polynomial-time many-one reduction to \twosatAcr{} with instance~$(X,\frml_1)$.
 \xcase{(ii)}{$d=0$} 
 Polynomial-time many-one reduction to \twosatAcr{} 
 with instance~$(X,\frml')$,
 wheret~$\frml'=\bigland_{i=1}^\tau \frml_i$,
 \xcase{(iii)}{$d=n$}
 Solve each of the $\tau$ instances~$(X,\frml_1),\dots,(X,\frml_\tau)$ of~\twosatAcr{} individually (Turing~reduction).
\end{proof}

\noindent
We will prove that the cases~(i) and~(ii) in~\cref{obs:msatlintime} are tight:
\msatTsc{} becomes \NP-hard if~$\tau\geq 2$ 
(\cref{ssec:twostages,ssec:allbutk}) 
or~$d=1$ 
(\cref{ssec:onechange}).
For the case~(iii) in~\cref{obs:msatlintime} the picture looks different:
we prove \msatTsc{} to be polynomial-time solvable if~$n-d\in\O(1)$~(\cref{sec:cnc}).

\subsection{From One to Two Stages}
\label{ssec:twostages}

In this section,
we prove that 
\msatTsc{} becomes \NP-hard if~$\tau\geq 2$.
In fact,
we prove the following.

\begin{theorem}%
 \label{thm:nphardnesstau}
\msatTsc{} is \NP-hard, 
even for two stages,
where the variables appear all negated in one
and all positive in the other stage. 
Moreover, \msatTsc{} 
\begin{enumerate}[(i)]
 \item is \W{1}-hard when parameterized by~$d$ even if~$\tau=2$,\label{hard:wone}
 \item admits no~$n^{o(d)\cdot f(\tau)}$-time algorithm for any function~$f$ unless~\ETHbreaks, and\label{hard:eth}
 \item admits no problem kernelization of size~$\O(n^{2-\eps}\cdot f(\tau))$ for any~$\eps>0$ and function~$f$,
 unless~\NPincoNPslashpoly{}.\label{hard:noqk}
\end{enumerate}
\end{theorem}

\noindent
We will reduce from the following \NP-hard problem:

\problemdef{Weighted 2-SAT}
{A set of variables~$X$, a 2-CNF~$\frml$ over~$X$, and an integer~$k$.}
{Is there satisfying \assment{} for~$\frml$ with at most~$k$ variables set true?}

\noindent
When parameterized by the number~$k$ of set-to-true variables,
\prob{Weighted 2-SAT} is \W{1}-complete~\cite{DowneyF99,FlumG06}.
Moreover,
\prob{Weighted 2-SAT} admits
no~$n^{o(k)}$-time algorithm unless~\ETHbreaks~\cite{ChenHKX06}
and no problem bikernelization of size~$\O(n^{2-\eps})$, $\eps>0$, unless~\NPincoNPslashpoly~\cite{DellM14}.

\begin{construction}
 \label{constr:wsat}
 Let~$(X,\frml,k)$ be an instance of \prob{Weighted 2-SAT},
 where~$\frml=\bigland_{i=1}^m C_i$.
 Construct~$\seqf=(\frml_1,\frml_2)$,
 where~$\frml_2\ceq \frml$ 
 and~$\frml_1\ceq \bigland_{x\in X}(\xneg{x})$ consists of~$n$ size-one clauses, 
 where each variable appears negated in one clause.
 Finally, set~$d\ceq k$.
 \cqed
\end{construction}

\begin{lemma}%
 \label{lem:wsat:corr}
 Let~$\I=(X,\frml,k)$ be an instance of \prob{Weighted 2-SAT},
 and let~$\I'=(X,\seqf,d)$ be an instance of \msatTsc{} obtained from~$\I$ using~\cref{constr:threesat}.
 Then,
 $\I$ is a \yes-instance if and only if $\I$' is a \yes-instance.
\end{lemma}

\begin{proof}
 \RD{}
 Let~$f$ be a satisfying \assment{} for~$\I$.
 We claim that~$f_1\colon X\to \{\false\}$,
 $x\mapsto\false$,
 and $f_2\ceq f$ is a solution to~$\I'$.
 Note that $f_1$ and~$f_2$ satisfy~$\frml_1$ and~$\frml_2$,
 respectively.
 Moreover,
 $|\{x\in X \mid f_1(x)\neq f_{2}(x)\}|=|\{x\in X \mid f_2(x)=\true\}|\leq d=k$.
 
 \LD{}
 Let~$(f_1,f_2)$ be a solution to~$\I'$.
 Note that~$f_1\colon X\to \{\false\}$.
 Since~$d=k$,
 there are at most~$k$ variables set to true by~$f_2$.
 Hence,
 $f_2$ is a satisfying \assment{} for~$\I$ with at most~$k$ variables set to true,
 and thus~$\I$ is \yes-instance.
\end{proof}

\begin{proof}[Proof of~\cref{thm:nphardnesstau}]
 \cref{constr:wsat} forms a polynomial-time many-one reduction to an instance with two stages with~$d=k$.
 Note that \prob{Weighted 2-SAT} remains \NP-hard if all literals are positive
 (e.g., via a reduction from~\prob{Vertex Cover}).
 Hence,
 \msatAcr{} is \NP-hard,
 even if the variables appear all negated in one
 and all positive in the other stage. 
 Moreover, 
 unless~\ETHbreaks,
 \msatAcr{} admits no~$n^{o(d)\cdot f(\tau)}$-time algorithm for any function~$f$ since no~$n^{o(k)}$-time algorithm exists for \prob{Weighted 2-SAT}~\cite{ChenHKX06}.
 As \cref{constr:wsat} also forms a \pt,
 \msatAcr{} is \W{1}-hard when parameterized by~$d$ even if~$\tau=2$.
 Moreover,
 \cref{constr:wsat} forms a \lpt{} 
 from \prob{Weighted 2-SAT} parameterized by~$|X|$
 to \msatAcr{} parameterized by~$n\cdot f(\tau)$ for any function~$f$.
 Hence,
 \msatAcr{} admits no problem kernel of size~$\O(n^{2-\eps}\cdot f(\tau))$ for any~$\eps>0$ and function~$f$,
 unless~\NPincoNPslashpoly.
\end{proof}

\begin{remark}
 \cref{thm:nphardnesstau}\eqref{hard:noqk} can be generalized to~\prob{Multistage $q$-SAT}:
 Instead from \prob{Weighted $2$-SAT},
 we reduce 
 (in an analogous way) 
 from~\prob{Weighted $q$-SAT} which admits no problem bikernelization of size~$\O(n^{q-\eps})$, $\eps>0$, unless~\NPincoNPslashpoly~\cite{DellM14}.
 Thus,
 unless~\NPincoNPslashpoly,
 \prob{Multistage $q$-SAT} admits no problem kernel of size~$\O(n^{q-\eps}\cdot f(\tau))$ for any~$\eps>0$ and function~$f$.
\end{remark}

\subsection{From Zero to One Allowed Change}
\label{ssec:onechange}

In this section,
we prove that 
\msatTsc{} becomes \NP-hard if~$d=1$ 
and the maximum number~$m$ of clauses in any stage is six.
In fact,
we prove the following.

\begin{theorem}
 \label{thm:nphardnessd}
\msatTsc{} is \NP-hard, 
even if the number of clauses in each stage is at most six and~$d=1$. 
Moreover, 
unless~\ETHbreaks,
\msatTsc{} admits no~$\O^*(2^{o(n)})$-time algorithm.
\end{theorem}

\begin{construction}
 \label{constr:threesat}
 Let~$(X,\frml)$ be an instance of \textsc{3-SAT},
 where~$\frml=\bigland_{i=1}^m C_i$ and each clause consists of exactly three literals.
 Let~$\ell_i^j$, $j\in\{1,2,3\}$, denote the literals in~$C_i$ for each~$i\in\set{m}$.
 Construct instance~$(X',\Phi,d)$ of \msatAcr{} as follows.
 First,
 construct~$X'\ceq X\cup B$, where~$B\ceq \{b_1,b_2,b_3\}$.
 Let
 \begin{align*} \frml_B &\ceq \cls{b_1}{b_2}\land \cls{b_1}{b_3} \land \cls{b_2}{b_3}, \text{ and} \\ 
 \frml_{\xneg{B}} &\ceq \cls{\xneg{b_1}}{\xneg{b_2}} \land \cls{\xneg{b_1}}{\xneg{b_3}} \land \cls{\xneg{b_2}}{\xneg{b_3}}.\end{align*}
 Next,
 construct~$\seqf\ceq (\frml_i,\dots,\frml_{2m})$ as follows.
 For each~$i\in\set{m}$,
 construct
 \begin{align*}
  \frml_{2i-1}\ceq \frml_{\xneg{B}}, 
  && \text{and} &&
  \frml_{2i}\ceq 
  \cls{\ell_1^i}{b_1} \land \cls{\ell_2^i}{b_2} \land \cls{\ell_3^i}{b_3}\land \frml_B
 \end{align*}
 Finally, set~$d\ceq 1$.
 \cqed
\end{construction}

\begin{observation}%
 \label{obs:onlybchanges}
 In every solution to an instance obtained from~\cref{constr:threesat},
 in each odd stage
 exactly two~$b_j$ are set to false
 and in each even stage
 exactly two~$b_j$ are set to true.
\end{observation}

\begin{proof}
 Clearly, 
 in every satisfying \assment{} for~$\frml_{2i-1}=\frml_{\xneg{B}}$,
 $i\in\set{m}$,
 at least two~$b_j$ are set to false.
 In every satisfying \assment{} for~$\frml_{2i}$,
 to satisfy~$\frml_B$,
 at least two variables from~$B$ must be set to true.
 As~$d=1$,
 exactly one of~$B$ being set to false in a satisfying \assment{}~$\frml_{2i-1}$ can be set to true in a satisfying \assment{} for~$\frml_{2i}$,
 which implies that already one of~$B$ must be set to true in~$\frml_{2i-1}$.
\end{proof}

\begin{lemma}%
 \label{lem:threesat:corr}
 Let~$\I=(X,\frml)$ be an instance of \prob{3-SAT},
 and let~$\I'=(X',\seqf,d)$ be an instance of \msatTsc{} obtained from~$\I$ using~\cref{constr:threesat}.
 Then,
 $\I$ is a \yes-instance if and only if $\I$' is a \yes-instance.
\end{lemma}

\begin{proof}
 \RD{}
 Let~$f\colon X\to \{\false,\true\}$ be a satisfying \assment{} for~$\frml$.
 We construct \assment{}s~$f_1,\dots,f_\tau\colon X'\to \{\false,\true\}$ as follows.
 Let~$f_i(x)=f(x)$ for all~$i\in\set{\tau}$ and all~$x\in X$.
 Next,
 for each~$i\in\set{m}$,
 $f_{2i}$ assigns exactly two variables from~$B$ to true such that~$\frml_{2i}$ is satisfied.
 This is possible since at least one clause from~$\frml_{2i}$ is already set to true by one true literal.
 It remains to show that for each~$i\in\set{m}$,
 there is an \assment{} of~$f_{2i-1}$ to the variables from~$B$ such that exactly two are set to false (in which case~$\frml_{2i-1}$ is satisfied),
 and~$|\{b\in B \mid f_{2i-2}(b)\neq f_{2i-1}(b)\}|\leq 1$ (if~$i=1$, interpret~$f_{2i-2}=f_{2i}$) and $|\{b\in B \mid f_{2i-1}(b)\neq f_{2i}(b)\}|\leq 1$.
 Since~$|B|=3$,
 there is a~$j\in\set{3}$ such that~$f_{2i-2}(b_j)=f_{2i}(b_j)=\true$.
 Set~$f_{2i-1}(b_j)=\true$ and~$f_{2i-1}(b_\ell)=\false$ for~$\ell\in\set{3}\setminus\{j\}$.
 Observe that~$|\{b\in B\setminus\{b_j\}\mid f_{2i-2}(b)\neq f_{2i-1}(b)\}| = |\{b\in B\setminus\{b_j\}\mid f_{2i-1}(b)\neq f_{2i}(b)\}|=1$,
 what we needed to show.
 
 \LD{}
 By~\cref{obs:onlybchanges},
 between every two consecutive stages,
 exactly one variable in~$B$ changes its true-false value.
 Hence,
 each variable from~$X$ is assigned the same value in each stage,
 i.e.,
 $f_i(x)=f_j(x)$ for every~$x\in X$ and every~$i,j\in\set{\tau}$.
 Let~$f:X\to\{\false,\true\}$ be the \assment{} of the variables in~$X$ with~$f(x)\ceq f_1(x)$ for all~$x\in X$. 
 Since in every even stage, 
 by~\cref{obs:onlybchanges},
 exactly one variable from~$B$ is set to false,
 at least one literal must be set to true.
 It follows that each clause in the \textsc{3-SAT} instance is satisfied by~$f$,
 that is,
 $f$ is a satisfying \assment{} for~$\I$.
 Thus,
 $\I$ is a \yes-instance.
\end{proof}

\begin{proof}[Proof of~\cref{thm:nphardnessd}]
 \cref{constr:threesat} forms a polynomial-time many-one reduction to an instance with $d=1$, 
 $m=6$,
 and~$n= |X|+3$.
 Hence,
 \msatAcr{} is \NP-hard,
 even if~$d=1$ and~$m=6$,
 and,
 unless~\ETHbreaks,
 admits no~$\O^*(2^{o(n)})$-time algorithm since no~$\O^*(2^{o(|X|)})$-time algorithm exists for \prob{3-SAT}~\cite{ChenHKX06}.
\end{proof}

\subsection{From All to All But~{\boldmath\(k\)} Allowed Changes}
\label{ssec:allbutk}

In this section,
we prove that \msatTsc{} is \W{1}-hard when parameterized by the lower bound~$n-d$ on the number of unchanged variables between any two consecutive stages.

\begin{theorem}
 \label{thm:allbutkNPh}
 \msatTsc{} is
 \W{1}-hard when parameterized by~$n-d$
 even if~$\tau=2$,
 and,
 unless~\ETHbreaks,
 admits no~$\Ost(n^{o(n-d)\cdot f(\tau)})$-time algorithm for any function~$f$.
\end{theorem}

\noindent
We reduce from the following \NP-hard problem:
\problemdef{Multicolored Independent Set (MIS)}
{An undirected, $k$-partite graph~$G=(V^1,\dots,V^k,E)$.}
{Is there an independent set~$S$ such that~$|S\cap V^i|=1$ for all~$i\in\set{k}$?}

\noindent
\prob{MIS} is
\W{1}-hard with respect to~$k$~\cite{FHRV09} and 
unless~\ETHbreaks,
there is no~$f(k)\cdot n^{o(k)}$-time algorithm~\cite{LokshtanovMS11}.

\begin{construction}
 \label{constr:allbutkNPh}
 Let~$\I=(G=(V^1,\dots,V^k,E))$ be an instance of~\prob{MIS} and let~$V\ceq V^1\uplus\dots\uplus V^k$, 
 $n\ceq |V|$,
 and~$V^i=\{v^i_1,\dots,v^i_{|V_i|}\}$ for all~$i\in\set{k}$.
 We construct an instance~$\I'=(X,(\frml_1,\frml_2),d)$ 
 with~$d\ceq n-k$ as follows.
 Let~$X:=X^1\cup\dots\cup X^k$ with~$X^i=\{x^i_j\mid v^i_j\in V_i\}$ for all~$i\in\set{k}$.
 Let for all~$i\in\set{k}$
 \begin{align*} 
 \phi^*_i &\ceq \bigland_{j,j'\in \set{|V^i|},\, j\neq j'} (\xneg{x_j^i}\lor \xneg{x_{j'}^i}), &&\text{and let}& \phi_E&\ceq \bigland_{\{v^i_j,v^{i'}_{j'}\}\in E} (\xneg{x^i_j}\lor \xneg{x^{i'}_{j'}}).
 \end{align*}
 Let
 \begin{align*} 
 \phi_1 &\ceq \bigland_{x\in X} (x) &&\text{and}&
 \phi_2 &\ceq \phi_E\land \bigland_{i\in\set{k}} \phi^*_i.
 \end{align*}
 This finishes the construction.
 \cqed
\end{construction}

\begin{lemma}%
 \label{lem:allbutkNPh:cor}
 Let~$\I=(G=(V^1,\dots,V^k,E))$ be an instance of \prob{MIS},
 and let~$\I'=(X,(\frml_1,\frml_2),d)$ be an instance of \msatTsc{} obtained from~$\I$ using~\cref{constr:allbutkNPh}.
 Then,
 $\I$ is a \yes-instance if and only if $\I$' is a \yes-instance.
\end{lemma}

\begin{proof}
 \RD{}
  Let~$S=\{v^1_{j_1},\dots,v^k_{j_k}\}\subseteq V$ be an independent set with~$S\cap V^i=\{v_{j_i}^i\}$ for all~$i\in\set{k}$.
  Let~$X_S\ceq \{x^1_{j_1},\dots,x^k_{j_k}\}$ be the variables in~$X$ corresponding to the vertices in~$S$.
  Let~$f_1\colon X\to \tfset,x\mapsto\true$ and~$f_2\colon X\to\tfset$ be defined as
  \[ f_2(x) = \begin{cases} \true,& \cif x\in X_S,\\ \false,& \otw. \end{cases} \] 
  Clearly,
  $f_1$ satisfies~$\frml_1$.
  Further, 
  observe that for each~$r\in\set{k}$,
  $f_2$ satisfies~$\phi_r^*$ since all variables from~$X^i$ except for~$x_{j_i}^i$ is set to~$\false$.
  Since~$S$ is an independent set,
  and only variables corresponding to vertices from~$S$ are set to true by~$f_2$,
  $f_2$ satisfies~$\phi_E$.
  It follows that~$f_2$ satisfies~$\phi_2$,
  and hence,
  $f=(f_1,f_2)$ is a solution for~$\I'$.
  
 \LD{}
  Let~$f=(f_1,f_2)$ be a solution to~$\I'$.
  Let~$S\ceq \{v_j^i\in V\mid f_2(v_j^i)=\true\}$.
  We claim that~$S$ is an independent set in~$G$ with~$|S\cap V^i|=1$ for all~$i\in\set{k}$.
  
  First,
  observe that~$S$ is an independent set in~$G$:
  Suppose not,
  then there are~$v^i_j,v^{i'}_{j'}\in S$ such that~$\{v^i_j,v^{i'}_{j'}\}\in E$.
  By construction,
  $f_2(x^i_j)=f_2(x^{i'}_{j'})=\true$.
  Since~$\phi_E$ contains the clause~$(\xneg{x^i_j}\lor \xneg{x^{i'}_{j'}})$,
  $f_2$ does not satisfy~$\phi_E$ 
  (and, thus, $\phi_2$),
  contradicting the fact that~$f$ is a solution.
  It follows that~$S$ is an independent set in~$G$.
  
  It remains to show that $|S\cap V^i|=1$ for all~$i\in\set{k}$.
  Observe that for all~$i\in\set{k}$,
  $|\{x\in X^i\mid f_2(x)=\true\}|\leq 1$.
  Suppose not,
  that is,
  there is some~$i\in\set{k}$ 
  such that there are~$x,y\in X^i$ with~$x\neq y$ and~$f_2(x)=f_2(y)=\true$.
  By construction,
  $\phi_i^*$ contains
  the clause~$(\xneg{x}\lor\xneg{y})$,
  which is evaluated to false under~$f_2$.
  This is a contradiction to the fact that~$f$ is a solution.
  
  Observe that for all~$i\in\set{k}$,
  $|\{x\in X^i\mid f_2(x)=\true\}|> 0$:
  By construction of~$\phi_1$,
  we know that~$f_1(x)=\true$ for all~$x\in X$.
  Since~$d=n-k$,
  there are at least~$k$ vertices set to true by~$f_2$.
  If for some~$i\in\set{k}$ we have that~$|\{x\in X^i\mid f_2(x)=\true\}|= 0$,
  then,
  by the pigeon-hole principle,
  there is an~$i'\in\set{k}$ with~$i\neq i'$ 
  such that~$|\{x\in X^i\mid f_2(x)=\true\}|\geq 2$,
  which yields a contradiction as discussed.
  Thus,
  $|S\cap V^i|=1$ for all~$i\in\set{k}$.
  It follows that
  $S$ is a solution to~$\I$.
\end{proof}

\begin{proof}[Proof of~\cref{thm:allbutkNPh}]
 \cref{constr:allbutkNPh} runs in polynomial time 
 and outputs an equivalent instance~(\cref{lem:allbutkNPh:cor}) with two stages and~$d=n-k$.
 As \cref{constr:wsat} also forms a \pt,
 \msatAcr{} is \W{1}-hard when parameterized by~$n-d$ even if~$\tau=2$.
 Moreover, 
 unless~\ETHbreaks,
 \msatAcr{} admits no~$n^{o(n-d)\cdot f(\tau)}$-time algorithm for any function~$f$ since no~$n^{o(k)}$-time algorithm exists for \prob{MIS}.
\end{proof}

\section{Fixed-Parameter Tractability Regarding the Number of Variables and {\boldmath\(m+n-d\)}}
\label{sec:fpt}

In this section,
we prove that \msatTsc{} is fixed-parameter tractable regarding the number of variables (\cref{ssec:fptn})
and regarding the parameter~$m+n-d$,
the maximum number of clauses over all input formulas and the minimum number of variables not changing between any two consecutive stages (\cref{ssec:fptmnd}).

\subsection{Fixed-Parameter Tractability Regarding the Number of Variables}
\label{ssec:fptn}

We prove that \msatTsc{} is fixed-parameter tractable regarding the number of variables.

\begin{theorem}%
 \label{thm:fptn}
 \msatTsc{} is solvable in~$\O(\min\{2^{n}n^d,4^n\}\cdot\tau\cdot(n+m))$ time.
\end{theorem}

\begin{proof}
 Let~$\I=(X,\frml,d)$ be an instance of~\msatAcr{} with~$\seqf=(\frml_1,\dots,\frml_\tau)$.
 Construct the digraph~$D$ with vertex set~$V=V^1\uplus\dots\uplus V^\tau \uplus \{s,t\}$ and arc set~$A$ as follows.
 Add two designated vertices~$s$ and~$t$ to~$V$.
 For each~$i\in\set{\tau}$,
 for every \assment{}~$f$ satisfying~$\frml_i$,
 add a vertex~$v_f^i$ to~$V^i$.
 Note that there are at most~$2^n$ \assment{}s,
 where we can test for each \assment{} whether it is satisfying in~$\O(n+m)$ time.
 Add the arc~$(s,v)$ for all~$v\in V^1$ and the arc~$(v,t)$ for all~$v\in V^\tau$.
 Moreover,
 for each~$i\in\set{\tau-1}$,
 add the arc~$(v_f^i,v_g^{i+1})$ if and only if~$|\{x\in X\mid f(x)\neq g(x)\}|\leq d$.
 This finishes the construction of~$D$.
 Note that~$|V^i|\leq 2^n$,
 and each vertex (except for~$s$) has outdegree at most~$\sum_{j=1}^d \binom{n}{j}\leq n^d$.
 Hence,
 $|A|\in \O(\min\{2^{n}n^d,4^n\}\tau)$.
 
 It is not difficult to see that~$D$ admits an~$s$-$t$ path
 if and only if
 $\I$ is a~\yes-instance 
 (see,
 e.g.,
 \cite{FluschnikNRZ19,FluschnikNSZ20,BredereckFK20}).
 Checking whether~$D$ admits an~$s$-$t$ path can be done in~$\O(|V|+|A|)$.
\end{proof}

\begin{remark}
  \label{rem:adaptfptn}
 \cref{thm:fptn} is asymptotically optimal regarding~$n$ unless \ETHbreaks{}~(\cref{thm:nphardnessd}).
Moreover,
 \cref{thm:fptn} is easily adaptable to~\prob{Multistage $q$-SAT} with~$q\geq 3$ as,
 for every~$q\geq 3$,
 the number of \assment{}s is~$2^n$ 
 and each is verifiable in linear time.
\end{remark}

\subsection{Fixed-Parameter Tractability Regarding {\boldmath\(m+n-d\)}}
\label{ssec:fptmnd}

We prove that \msatTsc{} is fixed-parameter tractable regarding the parameter~$m+n-d$.

\begin{theorem}
 \label{thm:fptmnd}
 \msatTsc{} is solvable in~$\O(4^{2(m+n-d)}\tau(n+m))$ time.
\end{theorem}

\noindent
To prove~\cref{thm:fptmnd},
we will show that either~\cref{thm:fptn} applies with~$n\leq 2(m+n-d)$
or the following.

\begin{lemma}%
 \label{lem:2md}
 \msatTsc{} solvable in~$\O(\tau(n+m))$ time if~$2m<d$.
\end{lemma}

\begin{proof}
 Let~$\I=(X,\frml,d)$ be an instance of~\msatAcr{} with~$\seqf=(\frml_1,\dots,\frml_\tau)$ on~$n$ variables and each formula contains at most~$m$ clauses.
 Due to~\cref{rrule:trivialno},
 we can safely assume that each formula of~$\seqf$ admits a satisfying \assment{}.
 Let $X_i\subseteq X$ be the set of variables appearing as literals in~$\frml_i$ for each~$i\in\set{\tau}$.
 Note that~$|X_i|\leq 2m$ for each~$i\in\set{\tau}$.
 Compute in linear time a satisfying~\assment{} $f_1\colon X\to \tfset$ for~$\frml_1$.
 Compute for each~$i\in\set[2]{\tau}$ in linear time
 a satisfying \assment{}~$f'_i\colon X_i\to\tfset$ for~$\frml_i$.
 Next, iteratively for~$i=2,\dots,\tau$, set for all~$x\in X$
 \begin{align*}
    f_i(x) = \begin{cases}
           f_i'(x),& \cif{}x\in X_i,\\
           f_{i-1}(x),& \cif x\in X\setminus X_i.
          \end{cases}
 \end{align*}
 Clearly,
 \assment{}~$f_i$ satisfies~$\phi_i$.
 Moreover,
 for all~$i\in\set[2]{\tau}$ it holds that
 $|\{x\in X\mid f_{i-1}(x)\neq f_{i}(x)\}|\leq |X_i|\leq 2m<d$,
 and hence~$(f_1,\dots,f_\tau)$ is a solution to~$\I$.
\end{proof}

\begin{proof}[Proof of~\cref{thm:fptmnd}]
 Let~$\I=(X,\frml,d)$ be an instance of~\msatAcr{} with~$\seqf=(\frml_1,\dots,\frml_\tau)$ on~$n$ variables and each formula contains at most~$m$ clauses.
 We distinguish how $2(m+n-d)$ relates to~$2n-d$.
 \begin{description}
  \item[\textnormal{\xcase{1}{$2(m+n-d)\geq 2n-d$}}]
    Since~$d\leq n$,
  it follows that~$2(m+n-d)\geq n$.
  Due to~\cref{thm:fptn},
  we can solve~$\I$ in $\O(\min\{2^{n}n^d,4^n\}\tau(n+m))\subseteq \O(4^{2(m+n-d)}\tau(n+m))$ time.
  \item[\textnormal{\xcase{2}{$2(m+n-d)< 2n-d$}}]
  We have that
 \ifarxiv{}\[ 2(m+n-d)< 2n-d \iff 2m<d.\] \else{}$2(m+n-d)< 2n-d \iff 2m<d$.\fi{}
 Due to~\cref{lem:2md},
 we can solve~$\I$ in~$\O(\tau(n+m))$ time.
 \end{description}
\end{proof}

\begin{remark}
 \cref{thm:fptmnd} can be adapted for~\mqsatTsc{} for every~$q\geq 3$,
 where~\cref{lem:2md} is restated for~$qm<d$ and
 we check for a satisfying \assment{} for each stage in~$\Ost(2^{qm})$ time.
 To adapt the proof of~\cref{thm:fptmnd},
 we then relate~$q(m+n-d)$ with~$qn-(q-1)d$ 
 and either employ the adapted \cref{thm:fptn} (see~\cref{rem:adaptfptn}),
 or the adapted \cref{lem:2md}.
\end{remark}

\section{XP Regarding the Number of Consecutive Non-Changes}
\label{sec:cnc}

We prove that~\msatTsc{} is in \XP{} when parameterized by the lower bound~$n-d$ on non-changes between consecutive stages,
the parameter ``dual''~to~$d$.

\begin{theorem}
 \label{thm:cnc}
 \msatTsc{} is solvable in~$\O(n^{4(n-d)+1}\cdot 2^{4(n-d)}\tau(n+m))$ time.
\end{theorem}

\noindent
Let~$\I=(X,\seqf=(\frml_1,\dots,\frml_\tau),d)$ be a fixed yet arbitrary instance with~$n$ variables.
Two partial \assment{}s~$f_Y\colon Y\to \tfset$ and~$f_Z\colon Z\to\tfset$ with~$Y,Z\subseteq X$
are called~\emph{compatible} if for all~$x\in Y\cap Z$ it holds that~$f_Y(x)=f_Z(x)$.
For two compatible assignments~$f_Y,f_Z$,
we denote by~
\begin{align*}
f_Y\cup f_Z&\ceq \begin{cases} f_Y(x),& x\in Y,\\
                              f_Z(x),& x\in Z\setminus Y.
                \end{cases}
\end{align*}

\noindent
With a similar idea as in the proof of~\cref{thm:fptn},
we will construct a directed graph with terminals~$s$ and~$t$ such that there is an~$s$-$t$ path in~$G$ if and only if~$\I$ is a \yes{}-instance.

\begin{construction}
 \label{constr:cnc}
 Given~$\I$,
 we construct a graph~$G=(V,E)$ with vertex set
 \ifarxiv{}\begin{align*}
  V &\ceq V^{1\to3}\cup V^{2\to3}\cup \dots \cup V^{\tau-2\to\tau}\cup\{s,t\},
 \end{align*}\else{}$V \ceq V^{1\to3}\cup V^{2\to4}\cup \dots \cup V^{\tau-2\to\tau}\cup\{s,t\},$\fi{}
 where for each~$Y,Z\in\binom{X}{n-d}$,
 we have that~$(f_Y,f_Z)\in V^{i\to i+2}$ if and only if~$f_Y,f_Z$ are compatible and each of~$\phi_i[f_Y]$,
 $\phi_{i+1}[f_Y\cup f_Z]$,
 and~$\phi_{i+2}[f_Z]$ is satisfiable,
 and the following arcs:
  (i)~$(s,v)$ for all~$v\in V^{1\to3}$,
  (ii)~$(v,t)$ for all~$v\in V^{\tau-2\to\tau}$, and
  (iii)~$((f_{Y},f_{Z}),(f_{Y'},f_{Z'}))\in V^{i\to i+2}\times V^{j\to j+2}$ if~$j=i+1$ and~$f_{Z}=f_{Y'}$ (implying that~$Z=Y'$).
 \cqed
\end{construction}

\begin{lemma}%
 \label{lem:cnc:rt}
 \cref{constr:cnc} computes a graph of size~$\O(n^{4(n-d)+1}\cdot 2^{4(n-d)}\tau)$ and can be done in~$\O(n^{4(n-d)+1}\cdot 2^{4(n-d)}\tau(n+m))$ time.
\end{lemma}

\begin{proof}
 To construct a set~$V^{i\to i+2}$,
 we compute each tuple~$(f_Y,f_Y)$ in~$\O(n^{2(n-d)}\cdot 2^{2(n-d)})$ time,
 and check whether they are compatible in~$\O(n+m)$ time,
 and whether each of~$\phi_i[f_Y]$,
 $\phi_{i+1}[f_Y\cup f_Z]$,
 and~$\phi_{i+2}[f_Z]$ is satisfiable,
 each in~$\O(n+m)$ time.
 Since for any~$f_Y,f_Z$ 
 we can check whether~$f_Y=f_Z$ in~$\O(n)$ time,
 we add the~$\O(n^{4(n-d)}\cdot 2^{4(n-d)})$ many arcs from~$V^{i\to i+2}$ to $V^{i+1\to i+3}$ in~$\O(n^{4(n-d)+1}\cdot 2^{4(n-d)})$ time.
 In total,
 $G$ can be constructed in~$\O(n^{4(n-d)+1}\cdot 2^{4(n-d)}\tau(n+m))$ time.
\end{proof}

\begin{lemma}%
 \label{lem:cnc:corr}
 Let~$\I$ be an instance of~\msatTsc{} and let~$G$ be the graph obtained from applying~\cref{constr:cnc} to~$\I$.
 Then,
 $\I$ is a \yes-instance if and only if~$G$ admits an~$s$-$t$~paths.
\end{lemma}

\begin{proof}
 \RD{}
  Let~$f=(f_1,\dots,f_\tau)$ be a solution to~$\I$.
  For each~$i\in\set{\tau-1}$,
  since~$|\{x\in X\mid f_i(x)\neq f_{i+1}(x)\}|\leq d$,
  there is a set~$Y_i\subseteq \{x\in X\mid f_i(x)= f_{i+1}(x)\}$ with~$|Y_i|=n-d$.
  Observe that~$v^f_i\ceq (\restr{f_i}{Y_i},\restr{f_{i+1}}{Y_{i+1}})\in V^{i\to i+2}$: 
  Clearly~$\frml_i[\restr{f_i}{Y_i}]$ 
  and~$\frml_i[\restr{f_{i+2}}{Y_{i+1}}]$ 
  satisfiable.
  Note that $\restr{f_i}{Y_i},\restr{f_{i+1}}{Y_{i+1}}$ are compatible
  since~$Y_i\cap Y_{i+1}\subseteq \{x\in X\mid f_i(x)=f_{i+1}(x)=f_{i+2}(x)\}$.
  Moreover,
  $\frml_{i+1}[\restr{f_i}{Y_i}\cup\restr{f_{i+1}}{Y_{i+1}}]$ is satisfiable
  since~$f_{i+1}(x)=\restr{f_i}{Y_i}\cup\restr{f_{i+1}}{Y_{i+1}}(x)$ for all~$x\in Y_i\cup Y_{i+1}$.
  It follows that there is an~$s$-$t$ path in~$G$ 
  with the arc sequence~$((s,v^f_1),(v^f_1,v^f_2),\dots,(v^f_{\tau-1},t))$.
  
 \LD{}
 Let~$P$ be an~$s$-$t$ path in~$G$.
 By construction of~$G$,
 $P$ contains~$s$,
 $t$, and from each~$V^{i\to i+2}$ exactly one vertex.
 Moreover,
 if arc~$((f_X,f_Y),(f_{X'},f_{Y'}))$ is contained in~$P$,
 then~$f_Y=f_{X'}$.
 Let~$(s,((f_{Y_i},f_{Y_{i+1}}))_{i=1}^{\tau-2},t)$ be the sequence of vertices in~$P$.
 We know that there exists an~$f_1'\colon X\setminus Y_1\to \tfset$ that satisfies~$\phi_1[f_{Y_1}]$,
 and hence~$f_1\ceq f_1'\cup f_{Y_1}$ satisfies~$\phi_1$.
 Moreover,
 we know that for all~$i\in\set[2]{\tau-1}$,
 there exists~$f_i'\colon X\setminus(Y_{i-1}\cup Y_i) \to \tfset$ that satisfies~$\phi_i[f_{Y_{i-1}}\cup f_{Y_{i}}]$,
 and hence~$f_i\ceq f_i'\cup f_{Y_{i-1}}\cup f_{Y_i}$ satisfies~$\phi_i$.
 Finally,
 we know that there exists an~$f_\tau'\colon X\setminus Y_{\tau-1}\to \tfset$ that satisfies~$\phi_\tau[f_{Y_{\tau-1}}]$,
 and hence~$f_\tau\ceq f_\tau'\cup f_{Y_{\tau-1}}$ satisfies~$\phi_\tau$.
 
 It remains to show that~$|\{x\in X\mid f_i(x)\neq f_{i+1}(x)\}|\leq d$ for all~$i\in\set{\tau-1}$.
 Note that~$f_i(x)=f_{i+1}(x)$ for all~$x\in Y_{i}$,
 and since~$|Y_i|=n-d$,
 the claim follows.
\end{proof}

\begin{proof}[Proof of~\cref{thm:cnc}]
  Given an instance~$\I=(X,\seqf=(\frml_1,\dots,\frml_\tau),d)$ of \msatAcr{},
  apply~\cref{constr:cnc}
  in~$\O(n^{4(n-d)+1}\cdot 2^{4(n-d)}\tau(n+m))$ time
  to obtain graph~$G$ with terminals~$s$ and~$t$ of size~$\O(n^{4(n-d)+1}\cdot 2^{4(n-d)}\tau)$ (\cref{lem:cnc:rt}).
  Return, 
  in time linear in the size of~$G$, 
  \yes{} if $G$ admits an~$s$-$t$ path,
  and \no{} otherwise (\cref{lem:cnc:corr}).
\end{proof}

\begin{remark}
 \cref{thm:cnc} is asymptotically optimal regarding~$n-d$ 
 unless~\ETHbreaks{}~(\cref{thm:allbutkNPh}).
Moreover,
 \cref{thm:cnc} does not generalize to~\mqsatTsc{} for~$q\geq 3$,
 as~\mqsatAcr{} is already~\NP-hard for one stage and hence for any number~$n-d$.
\end{remark}

\section{XP Regarding Number of Stages and Consecutive Changes}

In this section,
we prove that~\msatTsc{} is in~\XP{}
when parameterized by~$\tau+d$.

\begin{theorem}
 \label{thm:xptaud}
 \msatTsc{} is solvable in~$\O(n^{2\tau\cdot d}\cdot 2^{\tau\cdot d+1}\cdot \tau\cdot (n+m))$ time.
\end{theorem}

Let~$\I=(X,\seqf=(\frml_1,\dots,\frml_\tau),d)$ be a fixed yet arbitrary instance with~$\tau\cdot d<n$,
as otherwise~\cref{thm:fptn} applies.
On a high level,
our \cref{alg:xp} works as follows:
\begin{enumerate}[(1)]
 \item Guess $q\leq \tau\cdot d$ variables~$X'\subseteq X$ that will change over time. %
 \item Guess an initial \assment{} of the variables in~$X'$. %
 \item For each but the first stage,
 guess the at most~$\min\{q,d\}$ possible variables to change.  %
 \item Set the variables to the guessed true or false values,
 delete clauses which are set to true.
 \item Return \yes{} if the resulting instance with~$d=0$ is a \yes{}-instance (linear-time checkable).
 \item If the algorithm never (for all possible guesses) returned \yes, 
 then return~\no.
\end{enumerate}
\begin{algorithm}[t]
\ForEach(\tcp*[f]{$1+n^{\tau\cdot d}$ many}){$X'\subseteq X: |X'|\leq \tau\cdot d$}{\label{alg:xp:forloopx}
  \ForEach(\tcp*[f]{$2^{|X'|}$ many}){$f_1\colon X'\to \{\false,\true\}$}{\label{alg:xp:forloopfone}
  $\frml_1^*\gets\textbf{simplify}(\frml_1,f_1)$\label{alg:xp:simpcone}\;
    \ForEach(\tcp*[f]{$2^\tau |X'|^{\tau\cdot d}$ many}){$g_2, g_3 , \dots , g_{\tau} : g_i\in \calF(X')\ \forall i\in\set[2]{\tau}$}{\label{alg:xp:forloopg}
      \ForEach{$i\in\set[2]{\tau}$}{
      $f_i(x)\gets f_{i-1}(x)\xor g_i(x)$ $\forall x\in X'$;\quad $\frml_i^*\gets \textbf{simplify}(\frml_i,f_i)$;\label{alg:xp:f}}
      \If{$(X\setminus X',(\frml_1^*,\dots,\frml_\tau^*),0)$ is a \yes-instance of \msatAcr{}}{
      \Return{\yes}\tcp*[f]{decidable in linear~time (\cref{obs:msatlintime})}\label{alg:xp:yes}
      }
    }
  }
}
\Return{\no}
\caption{\XP-algorithm on input instance~$(X,\frml,d)$. 
}
\label{alg:xp}
\end{algorithm}

\noindent
For any~$X'\subseteq X$,
define the set of all \assment{}s to variables of~$X'$ with at most~$\min\{|X'|,d\}$ true values by
\ifarxiv{}\begin{align*}
 \calF(X')\ceq \bigl\{f\colon X'\to \{\false,\true\}\bigm\vert |\{x\in X'\mid f(x)=\true\}|\leq \min\{|X'|,d\}\bigr\}.
\end{align*}
\fi{}
With the next two lemmas,
we prove that \cref{alg:xp} is correct and runs in \XP{}-time regarding~$\tau+d$.

\begin{lemma}%
 \label{lem:algyesIyes}
 \cref{alg:xp} returns~\yes{}
 if and only if
 the input instance is a \yes-instance.
\end{lemma}

\begin{proof}
  \RD{}
  If \cref{alg:xp} returns~\yes,
  then for some~$X'\subseteq X$,
  and some~$f_1,\dots,f_\tau$ that simplified~$\frml_1,\dots,\frml_\tau$ to~$\frml_1^*,\dots,\frml_\tau^*$,
  instance~$\I^*\ceq (X\setminus X',(\frml_1^*,\dots,\frml_\tau^*),0)$ is a \yes-instance of \msatAcr{}.
  Let~$f_1^*,\dots,f_\tau^*\colon X\setminus X'\to \tfset$ be a solution to~$\I^*$.
  Let~$h_1,\dots,h_\tau$ be defined as~$h_i(x) \ceq f_i(x)$ if~$x\in X'$, 
  and~$h_i(x)\ceq f_i^*(x)$ otherwise,
  i.e., if~$x\in X\setminus X'$.
  We claim that~$(h_1,\dots,h_\tau)$ is a solution to~$\I$.
  Observe that~$h_i$ satisfies~$\frml_i$ for each~$i\in\set{\tau}$.
  Moreover, 
  for each~$i\in\set{\tau-1}$ we have
  $|\{x\in X\mid h_i(x)\neq h_{i+1}(x)\}|=|\{x\in X'\mid g_{i+1}(x)=\true\}|\leq d$.

  \LD{}
  Let~$h=(h_1,\dots,h_\tau)$ be a solution to~$\I$.
  Let~$X'\subseteq X$ with~$|X'|\leq \tau\cdot d$ the set of all variables which change at least once over the stages their true-false value.
  \cref{alg:xp} guesses~$X'$ in~\autoref{alg:xp:forloopx}.
  Let~$f=(f_1,\dots,f_\tau)$ be such that for each~$i\in\set{\tau}$,
  $f_i\colon X'\to\tfset$ is~$h_i$ restricted to the variables in~$X'$.
  In~\autoref{alg:xp:forloopfone},
  \cref{alg:xp} guesses~$f_1$.
  Since~$h$ is a solution to~$\I$,
  we know that~$|\{x\in X\mid h_i(x)\neq h_{i+1}(x)\}|=|\{x\in X\mid f_i(x)\neq f_{i+1}(x)\}|\leq \min\{|X'|,d\}$ for each~$i\in\set{\tau-1}$.
  It follows that for each~$i\in\set[2]{\tau}$ there exists a~$g_i\in \calF(X')$ such that~$f_i(x)=f_{i-1}(x)\xor g_i(x)$.
  \cref{alg:xp} guesses~$g_2,\dots,g_\tau$ in~\autoref{alg:xp:forloopg},
  and finds~$f$ in~\autoref{alg:xp:f}.
  Let~$(\frml_1^*,\dots,\frml_\tau^*)$  be the formulas~$(\frml_1,\dots,\frml_\tau)$ simplified according to~$f$,
  as done by \cref{alg:xp} in~\autoref{alg:xp:simpcone} and~\autoref{alg:xp:f}.
  Since~$h$ is a solution,
  $f'=(f_1',\dots,f_\tau')$ where for each~$i\in\set{\tau}$,
  $f_i'\colon X\setminus X'\to\tfset$ is~$h_i$ restricted to the variables in~$X\setminus X'$,
  is a solution to~$(X\setminus X',(\frml_1^*,\dots,\frml_\tau^*),0)$.
  Hence,
  $(X\setminus X',(\frml_1^*,\dots,\frml_\tau^*),0)$ is a \yes-instance,
  and consequently \cref{alg:xp} returns~\yes{} in~\autoref{alg:xp:yes}.
\end{proof}

\begin{lemma}%
 \label{lem:algrt}
 \cref{alg:xp} runs in~$\O(n^{2\tau\cdot d}\cdot 2^{\tau\cdot d+1}\tau(n+m))$ time.
\end{lemma}

\begin{proof}
 The running time~$T(\I)$ is~$T(\I)\leq (1+n^{\tau\cdot d})\cdot T_1(\I)$,
 where~$T_1(\I)$ is the worst-case running time inside the first for-loop
 (\autoref{alg:xp:forloopfone} to~\autoref{alg:xp:yes}).
 Analogously,
 we have~$T_1(\I)\leq 2^{\tau\cdot d}\cdot T_2(\I)$,
 and~$T_2(\I)\in \O(n+m) + (1+(\tau\cdot d)^d)^{\tau-1}\cdot T_3(\I)$.
 Now,
 $T_3(\I)\in \O(\tau(n+m))$,
 as \autoref{alg:xp:f} can be done in~$\O(n+m)$~time with~$(\tau-1)$ executions of this line,
 and checking the if-condition for~\autoref{alg:xp:yes} can be done in~$\O(\tau(n+m))$ time.
 We arrive at
 \begin{align*}
  T(\I) &\in \O((1+n^{\tau\cdot d})\cdot 2^{\tau\cdot d}\cdot ((n+m)+(1+\tau\cdot d)^{\tau\cdot d}\cdot \tau(n+m))) \\
  & \subseteq \O(n^{2\tau\cdot d}\cdot 2^{\tau\cdot d+1}\cdot \tau(n+m))
  \qedhere
 \end{align*}
\end{proof}

\noindent
We are set to prove the main result from this section.

\begin{proof}[Proof of~\cref{thm:xptaud}]
 Let~$\I=(X,\seqf=(\frml_1,\dots,\frml_\tau),d)$ be an instance of~\msatAcr{}
 with $n$~variables and at most~$m$ clauses in each stage's formula.
 If~$\tau\cdot d\geq n$,
 then,
 by~\cref{thm:fptn},
 we know that \msatAcr{} is solvable in~$\O(2^{2\tau\cdot d}\cdot \tau(n+m))$ time.
 Otherwise,
 if~$\tau\cdot d<n$,
 then
 \cref{alg:xp} runs in~$\O(n^{2\tau\cdot d}\cdot 2^{\tau\cdot d+1}\tau(n+m))$ time 
 (\cref{lem:algrt})
 and correctly decides~$\I$ (\cref{lem:algyesIyes}).
\end{proof}

\begin{remark}
 \cref{thm:xptaud} is asymptotically optimal regarding~$d$ unless \ETHbreaks{}~(\cref{thm:nphardnesstau}).
Moreover,
 \cref{thm:xptaud} is not adaptable to~\prob{Multistage $q$-SAT} with~$q\geq 3$ unless~$\classP=\NP$
 since \prob{Multistage $q$-SAT} with~$q\geq 3$ is \NP-hard even with~$\tau+d\in \O(1)$.
\end{remark}

\section{Efficient and Effective Data Reduction}

In this section,
we study efficient and provably effective data reduction for \msatTsc{} in terms of problem kernelization.
We focus on the parameter combinations~$n+m$, 
$n+\tau$, 
and~$m+\tau$.
We prove that no problem kernelization of size polynomial in~$n+m$ exists unless~\NPincoNPslashpoly{} (\cref{ssec:nopkmn}),
and that a problem kernelization of size quadratic in~$m+\tau$ and
of size cubic in~$n+\tau$ exists (\cref{ssec:ppk}).
Finally,
we prove that no problem kernel of size truly subquadratic in~$m+\tau$ exists unless~\NPincoNPslashpoly{}\ifarxiv{}~(\cref{sssec:nosubqk})\else{}\fi{}.

\subsection{No Time-Independent Polynomial Problem Kernelization}
\label{ssec:nopkmn}

When parameterized by~$n+m$,
efficient and effective data reduction appears unlikely.

\begin{theorem}%
 \label{thm:nopkmn}
 Unless~$\NPincoNPslashpoly$,
 \msatTsc{} admits no problem kernel of size polynomial in~$n^{f(m,d)}$,
 for any function~$f$ only depending on~$m$ and~$d$.
\end{theorem}

\noindent
We will prove \cref{thm:nopkmn} via an AND-composition~\cite{BodlaenderDFH09,BodlaenderJK14},
that is,
we prove that given $t$ instances of \msatTsc{}, 
each with~$d=1$ and the same number of variables and stages,
we can compute in polynomial time an instance of \msatTsc{}
such that all input instances are \yes{} if and only if the output instance is \yes,
and the number of variables and the maximum number of clauses in one stage does not exceed those from all input instances.
Drucker~\cite{Drucker15} proved that if a parameterized problem admits an AND-composition  from an~\NP-hard problem,
then it admits no polynomial problem kernelization,
unless~\NPincoNPslashpoly.

\begin{construction}
 \label{constr:andcroco}
 Let~$\I_1,\dots,\I_t$ be $t$ instances of \msatAcr{} with~$d=1$,
 $m=6$,
 $n$~variables,
 and $\tau$~stages,
 where~$\I_i=(X^i,\seqf^i,d)$ with~$X^i=\{x_1^i,\dots,x^i_n\}$
 and~$\seqf^i=(\frml^i_1,\dots,\frml^i_\tau)$.
 Construct the instance~$\I\ceq (X,\seqf,d)$ as follows.
 Construct the set~$X=\{x_1,\dots,x_n\}$ of variables,
 and identify~$x_j$ with~$x_j^i$ for each~$i\in\set{t}$, $j\in\set{n}$.
 In a nutshell,
 we construct the sequence of formulas by chaining up the input instances' formulas,
 and add~$n$ stages between any two consecutive instances each consisting of the always-true formula~$\cls{x_1}{\xneg{x_1}}$---these ensure a reconfiguration of the last \assment{} to the initial \assment{} of the subsequent instance.
 Formally,
 construct~$\seqf=(\frml_1,\dots,\frml_{t\cdot(\tau+n)})$ as follows.
 For all~$i\in\set{t}, j\in\set{\tau+n}$, 
 set (where $S(i)\ceq (i-1)\cdot(\tau+n)$)
 \[ \frml_{S(i)+j} \ceq \begin{cases} \frml^i_j,& \cif 1\leq j\leq \tau, \\ (x_1\lor \xneg{x_1}),& \cif \tau+1\leq j\leq \tau+n. \end{cases} \]
 Finally,
 set~$d=1$.
 \cqed
\end{construction}

\begin{lemma}%
 \label{lem:andcrocoopt}
 Let~$\I_1,\dots,\I_t$ be $t$ instances of \msatTsc{} with~$d=1$,
 $m=6$, 
 $n$~variables,
 and $\tau$ stages,
 and let~$\I$ be the instance obtained from~\cref{constr:andcroco}.
 Then,
 each~$\I_i$ is a \yes-instance if and only if~$\I$ is a \yes-instance.
\end{lemma}

\begin{proof}
 \LD{}
 Let~$(f_1,\dots,f_{t(\tau+n)})$ be a solution to~$\I$.
 It is not difficult to see that,
 for each~$i\in\set{t}$,
 the sequence~$(f_{S(i)+1},\dots,f_{S(i)+\tau})$ is a solution to~$\I_i$.
 
 \RD{}
 For each~$i\in\set{t}$,
 let~$(f^i_1,\dots,f^i_\tau)$ denote a solution for~$\I_i$.
 We construct a solution~$f=(f_1,\dots,f_{t(\tau+n)})$ for~$\I$ as follows.
 For each~$i\in\set{t}$ and~$j\in\set{\tau}$,
 set~$f_{S(i)+j} \ceq f^i_j$.
 For each~$i\in\set{t-1}$,
 we define $f_{S(i)+\tau+1},\dots,f_{S(i)+\tau+n}$ iteratively as follows.
 For~$j=1,\dots,n$,
 let~
 \[ f_{S(i)+\tau+j}(x) \ceq \begin{cases} f_{S(i)+\tau+(j-1)}(x), & \cif x\in X\setminus\{x_j\}, \\ f_{S(i+1)+1}(x),& \cif x=x_j. \end{cases}\]
 Observe that for each~$j\in\set{n}$,
 it holds true that~$|\{x\in X\mid f_{S(i)+\tau+(j-1)}(x)\neq f_{S(i)+\tau+j}(x)\}|\leq 1$,
 and that~$f_{S(i)+\tau+n}=f_{S(i+1)+1}$.
\end{proof}

\begin{proof}[Proof of~\cref{thm:nopkmn}]
  \cref{constr:andcroco} forms an AND-composition 
  (\cref{lem:andcrocoopt})
  from an \NP-hard special case of~\msatAcr{} 
  (\cref{thm:nphardnessd})
  to \msatAcr{} when parameterized by~$n+m$,
  in fact,
  mapping~$m$ and~$d$ to a constant.
  Thus,
  due to~Drucker~\cite{Drucker15},
  \msatAcr{}
  admits no problem kernelization of size polynomial in~$n^{f(m,d)}$
  for any function~$f$ only depending on~$m$ and~$d$.
\end{proof}

\begin{remark}
 Due to~\cref{thm:fptn},
 \msatTsc{} yet admits a problem kernel of size~$2^{\O(n)}$.
\end{remark}

\subsection{Polynomial Problem Kernelizations}
\label{ssec:ppk}

We prove problem kernelizations of size polynomial in~$n+\tau$ and~$m+\tau$.

\begin{theorem}
 \label{thm:qaudkermntau}
 \msatTsc{} admits a linear-time computable problem kernelization of size~$\O(n^2\tau)$ and of size~$\O(m\cdot\tau)$.
\end{theorem}

\noindent
We employ the following two immediate reduction rules
(each is clearly correct and applicable in linear time):

\begin{rrule}
 \label{rr:clauseclones}
 In each stage,
 delete all but one appearances of a clause in the formula.
\end{rrule}

\begin{rrule}
 \label{rr:ghostvars}
 Delete a variable that appears in no stage's formula as a literal.
\end{rrule}

\begin{proof}[Proof of~\cref{thm:qaudkermntau}]
 Observe that there are at most~$N \ceq 2n+\binom{2n}{2}\in \O(n^2)$ many pairwise different clauses.
 After exhaustively applying~\cref{rr:clauseclones},
 we have~$m\leq N\in \O(n^2)$.
 After exhaustively applying~\cref{rr:ghostvars},
 it follows that for each variable, 
 there is at least one clause,
 and hence,
 $n\leq  2\cdot m\cdot\tau$.
\end{proof}

\begin{remark}
 \cref{thm:qaudkermntau} adapts easily to \prob{Multistage $q$-SAT}.
 Herein,
 the problem kernel sizes are~$\O(n^q\cdot \tau)$ and~$\O(q\cdot m\cdot\tau)$.
\end{remark}

\noindent
Subsequently,
we prove that a linear kernel
appears unlikely.

\ifarxiv{}
  \subsubsection{No Subquadratic Problem Kernelization}
  \label{sssec:nosubqk}
\fi{}

\begin{theorem}
 \label{thm:nolinkermntau}
 Unless \NPincoNPslashpoly{},
 \msatTsc{} admits no problem kernel of size~$\O((m+n+\tau)^{2-\eps})$ for any~$\eps>0$.
\end{theorem}

\noindent
To prove~\cref{thm:nolinkermntau},
we show that there is a \lpt{} from \prob{Vertex Cover} parameterized by~$|V|$
to~\msatTsc{} parameterized by~$n+m+\tau$.

\begin{construction}
 \label{constr:nolinkermntau}
 Let~$\I=(G,k)$ with~$G=(V,E)$ be an instance of \prob{Vertex Cover}.
 Denote the vertices~$V=\{v_1,\dots,v_n\}$.
 We construct the instance~$\I'=(X,\seqf,d)$ of \msatAcr{} with~$d=k$ and~$\seqf=(\frml_0,\frml_1,\dots,\frml_n)$ as follows.
 Let~$X=X_V\cup B$ with~$X_V=\{x_i\mid v_i\in V\}$ and~$B=\{b_1,\dots,b_k\}$.
 Let
 \begin{align*}
\frml_0 &\ceq \bigland_{i=1}^n (\xneg{x_i}) \land \bigland_{j=1}^k (\xneg{b_j}) \text{~~~and} && \\
 \frml_i &\ceq \bigland_{\{v_{i},v_j\}\in E} \cls{x_{i}}{x_j} \land 
 \begin{cases} 
    \bigland_{j=1}^k (b_j) & \cif{}i\bmod 2=0,  \\ 
    \bigland_{j=1}^k (\xneg{b_j}) & \cif{}i\bmod 2=1,  
 \end{cases} 
 && \forall i\in\set{n}.
 \end{align*}
 Note that~$\tau+m+|X|\in \O(n)$,
 since each vertex degree is at most~$n-1$.
 \cqed
\end{construction}

\begin{lemma}%
 \label{lem:nolinkermntau}
 Let~$\I=(G,k)$ be an instance of \prob{Vertex Cover},
 and let~$\I'=(X',\seqf',d)$ be the instance of \msatTsc{} obtained from~$\I$ using~\cref{constr:nolinkermntau}.
 Then,
 $\I$ is a \yes-instance if and only if $\I'$ is a \yes-instance.
\end{lemma}

\begin{proof}
 \RD{}
 Let~$V'\subseteq V$ be a size-at-most-$k$ vertex cover of~$G$.
 Let~$X_W\ceq \{x_i\in X_V\mid v_i\in W\}$.
 Define~$f_0\colon X\to \tfset$ such that~$f_0(x)=\false$ for all~$x\in X$.
 Define~$f_1,\dots,f_n\colon X\to \tfset$ and~$f^*\colon X_V\to\tfset$ as
 \begin{align*} 
 f_i(x) = \begin{cases} f^*(x), & \cif x\in X_V, \\ \true, & \cif x\in B\text{ and } i\bmod 2=0, \\ \false , & \cif x\in B\text{ and } i\bmod 2=1,\end{cases} 
 \ \ \text{ where }\ \
 f^*(x) = \begin{cases} \true, & \cif x\in X_W, \\ \false, & \cif x\in X_V\setminus X_W.\end{cases} 
 \end{align*}
 Observe that~$|\{x\in X\mid f_0(x)\neq f_1(x)\}|=|\{x\in X_V\mid f^*(x)=\true\}|=|X_W|=|W|\leq k$. 
 Moreover,
 for each~$i\in\set{n-1}$,
 we have that
 $|\{x\in X\mid f_i(x)\neq f_{i+1}(x)\}|=|B|=k$.
 It is not difficult to see that~$f_i$ satisfies~$\frml_i$ for each~$i\in\set[0]{\tau}$.
 Hence,
 $(f_0,f_1,\dots,f_n)$ is a solution to~$\I'$.
 
 \LD{}
 Let~$f=(f_0,f_1,\dots,f_n)$ be a solution to~$\I'$.
 By construction of~$\frml_0$,
 it must hold that~$f_0(x)=\false$ for all~$x\in X$.
 Moreover,
 by construction of~$\frml_1$,
 we know that~$f_1(x)=\false$ for all~$x\in B$, and hence
 $X'\ceq \{x\in X_V\mid f_0(x)\neq f_1(x)\}=\{x\in X_V\mid f_1(x)=\true\}$ has~$|X'|\leq k$.
 Since for each~$i\in\set{n-1}$,
 we have that~$\{x\in X\mid f_i(x)\neq f_{i+1}(x)\}=B$ by construction,
 we know that for each~$i,j\in\set{n}$ it holds true that~$f_i(x)=f_j(x)$ for all~$x\in X_V$.
 We claim that~$W=\{v_i\in V\mid x_i\in X'\}$ is a size-at-most-$k$ vertex cover of~$G$.
 We know that~$|W|=|X'|\leq k$.
 Suppose towards a contradiction that there is an edge~$\{v_i,v_j\}\in E$ disjoint from~$W$.
 This implies that~$f_i(x_i)=f_i(x_j)=\false$.
 By construction,
 $\frml_i$ contains the clause~$\cls{x_i}{x_j}$,
 which is not satisfied by~$f_i$.
 This contradicts the fact that~$f$ is a satisfying \assment{}.
 It follows that~$W$ is a size-at-most-$k$ vertex cover of~$G$,
 and thus,
 $\I$ is a \yes-instance.
\end{proof}

\begin{proof}[Proof of~\cref{thm:nolinkermntau}]
  \cref{constr:nolinkermntau} is a \lpt{} (\cref{lem:nolinkermntau})
  such that~$\tau+m+|X|\in \O(|V|)$.
  Since \prob{Vertex Cover} admits no problem bikernelization of size~$\O(|V|^{2-\eps})$, $\eps>0$~\cite{DellM14},
  the statement follows.
\end{proof}

\begin{remark}
 \cref{thm:nolinkermntau} can be easily adapted to~\prob{Multistage $q$-SAT}
 when taking~\prob{$q$\nobreakdash-Hitting Set} as source problem~\cite{DellM14},
 ruling out problem kernelizations of size~$\O((n+m+\tau)^{q-\eps})$, $\eps>0$ (unless~\NPincoNPslashpoly).
\end{remark}

\section{Conclusion}

While \prob{2-SAT} is linear-time solvable,
its multistage model~\msatTsc{} is intractable in even surprisingly restricted cases.
This is also reflected by the fact that several of our direct upper bounds are already asymptotically optimal.
By our results,
the most interesting difference between \msatTsc{} and~\mqsatTsc{}, 
with $q\geq 3$,
is that the former is efficiently solvable if the numbers of stages and allowed consecutive changes are constant,
which is not the case for the latter (unless~$\classP=\NP)$.
Finally,
our results show that exact solutions are far from practical,
waving the path for randomized or heuristic approaches.

{
\small
\bibliography{ms2sat-bib}
}

\end{document}